\documentclass[english]{amsart}
\usepackage{mathptmx}
\usepackage[T1]{fontenc}
\usepackage{amsmath}
\usepackage{amsthm}
\usepackage{amssymb}
\usepackage{graphicx}

\makeatletter
\numberwithin{equation}{section}
\numberwithin{figure}{section}
\theoremstyle{plain}
\newtheorem{thm}{\protect\theoremname}
\theoremstyle{definition}
\newtheorem{problem}[thm]{\protect\problemname}
\theoremstyle{definition}
\newtheorem{example}[thm]{\protect\examplename}
\theoremstyle{remark}
\newtheorem{rem}[thm]{\protect\remarkname}
\theoremstyle{plain}
\newtheorem{lyxalgorithm}[thm]{\protect\algorithmname}

\makeatother

\usepackage{babel}
\providecommand{\algorithmname}{Algorithm}
\providecommand{\examplename}{Example}
\providecommand{\problemname}{Problem}
\providecommand{\remarkname}{Remark}
\providecommand{\theoremname}{Theorem}

\begin{document}
\title{Computing Limits of Quotients of Multivariate Real Analytic Functions}
\author{Adam Strzebo\'nski}
\address{Wolfram Research Inc., 100 Trade Centre Drive, Champaign, IL 61820,
U.S.A. }
\email{adams@wolfram.com}
\date{January 31, 2021}
\begin{abstract}
We present an algorithm for computing limits of quotients of real
analytic functions. The algorithm is based on computation of a bound
on the \L ojasiewicz exponent and requires the denominator to have
an isolated zero at the limit point. 
\end{abstract}

\maketitle

\subsection*{Keywords:} Multivariate function limit, symbolic limit computation.

\section{Introduction}

Computation of limits is one of the basic problems of computational
calculus. In the univariate case computing limits of rational functions
is easy, and the state of the art limit computation algorithms \cite{G,SSH}
are applicable to large classes of functions. In the multivariate
case computing limits of real rational functions is a nontrivial problem
that has been a subject of recent research \cite{AKM,CMV,VHC,XZ,XZZ}.
In \cite{S15} we compared five methods for computation of limits
of real rational functions based on the Cylindrical Algebraic Decomposition
(CAD) algorithm. Here we extend the methods to computation of limits
of quotients of multivariate analytic functions.

The limit of a real function may not exist, however the lower limit
and the upper limit always exist. A weak version of the limit computation
problem consists of deciding whether the limit exists and, if it does,
finding the value of the limit. A strong version consists of finding
the values of the lower limit and the upper limit. 

Let us state the problems precisely. Denote $x=(x_{1},\ldots,x_{n})$,
$\bar{\mathbb{R}}=\mathbb{R}\cup\{-\infty,\infty\}$. Let $U\subseteq\mathbb{R}^{n}$
be an open set and let $\mathcal{A}(U)$ denote the set of analytic
functions on $U$. Let $g\in\mathcal{A}(U)$, $h\in\mathcal{A}(U)\setminus\{0\}$,
$D=\{u\in U\::\:h(u)\neq0\}$, $f:D\ni u\rightarrow\frac{g(u)}{h(u)}\in\mathbb{R}$,
and let $c\in U$.
\begin{problem}
\label{WeakLimit}Find $l\in\bar{\mathbb{R}}$ such that $l=\lim_{u\rightarrow c}f(u)$
or prove that such $l$ does not exist.
\end{problem}

\begin{problem}
\label{StrongLimit}Find $l_{1},l_{2}\in\bar{\mathbb{R}}$ such that
$l_{1}=\liminf_{u\rightarrow c}f(u)$ and $l_{2}=\limsup_{u\rightarrow c}f(u)$. 
\end{problem}

\begin{example}
Let $g=\sin(x^{2}+y^{2}+z^{2})$ and $h=3-\cos x-\cos y-\cos z$.
Then
\[
\lim_{(x,y,z)\rightarrow0}\frac{g(x,y,z)}{h(x,y,z)}=2
\]
\end{example}

\begin{example}
Let $g=\sin(xy)$ and $h=\cos x+\cos y-2$. Then
\begin{eqnarray*}
\liminf_{(x,y)\rightarrow0}\frac{g(x,y)}{h(x,y)} & = & -1\\
\limsup_{(x,y)\rightarrow0}\frac{g(x,y)}{h(x,y)} & = & 1
\end{eqnarray*}
 and $\lim_{(x,y)\rightarrow0}\frac{g(x,y)}{h(x,y)}$ does not exist.
\end{example}

\begin{figure}
\centering
\includegraphics[width=0.7 \textwidth, trim = 0mm 0mm 0mm 0mm, clip]{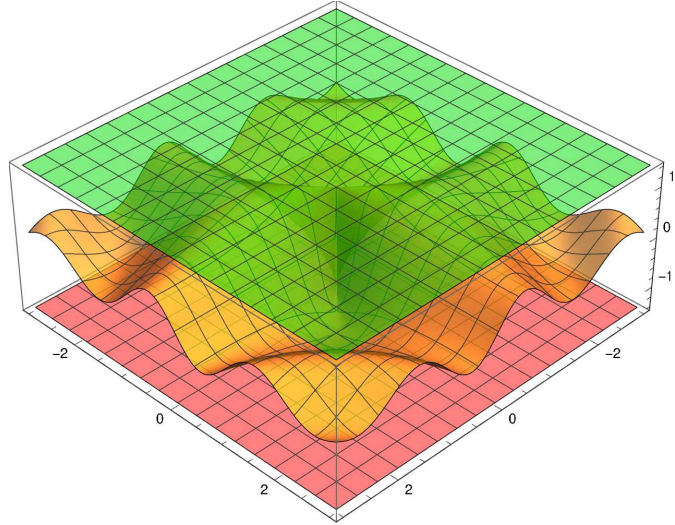}
\caption{\label{fig} The lower and the upper limit of 
$z=\frac{\sin(xy)}{\cos x+\cos y-2}$ at $0$. }
\end{figure}

Recently two algorithms partially solving Problem \ref{WeakLimit}
for rational functions have been proposed. The algorithm presented
in \cite{XZ} solves a modified version of the problem, namely it
decides whether the limit exists and is finite. The negative answer
includes both the case when the limit does not exist and the case
when the limit exists and is infinite. The algorithm uses Wu's elimination
method, rational univariate representations, and requires adjoining
two infinitesimal elements to the field. The algorithm presented in
\cite{AKM} (which generalizes algorithms of \cite{CMV,VHC}) solves
Problem \ref{WeakLimit} under the additional assumption that $c$
is an isolated zero of $h$. The authors use the theory of Lagrange
multipliers to reduce the problem to computing the limit along a real
algebraic set, and solve the reduced problem using regular chains
methods.

In \cite{S15} we presented five methods based on the CAD algorithm
that solve both Problem \ref{WeakLimit} and Problem \ref{StrongLimit}
for arbitrary rational functions. In this note we describe an algorithm
that reduces computation of limits of quotients of multivariate analytic
functions to computation of limits of rational functions. The algorithm
can be combined with any of the algorithms described in \cite{S15}
to solve both Problem \ref{WeakLimit} and Problem \ref{StrongLimit}
for quotients of analytic functions.

\section{The Algorithm}

Let $f=\frac{g}{h}$, where $g$ and $h$ are analytic functions in
a neighbourhood of $c\in\mathbb{R}^{n}$. Without loss of generality
we will assume that $c=0$. If $h(0)\neq0$ then $f$ is continuous
at $0$, and hence the limit can be obtained by evaluation. Therefore
w.l.o.g. we will assume that $h(0)=0$. The only computability assumption
we make about functions $g$ and $h$ is that for any $d\in\mathbb{N}$
we can compute the Taylor polynomials $T_{d}g$ and $T_{d}h$ of total
degree $d$. This is a typical case in a computer algebra system,
where $g$ and $h$ are given as expressions obtained by composing
analytic functions implemented in the system. Taylor polynomials of
arbitrary degree can be readily computed, but for instance algorithms
for testing whether an expression represents a function that is identically
zero may not be available.
\begin{rem}
\label{limitation}This specification is not sufficient to always
solve Problem \ref{WeakLimit} or \ref{StrongLimit}. For instance,
let $\varphi_{k}=x^{2}+y^{2k}$, let $\varphi_{\infty}=x^{2}$, and
suppose that $g$ and $h$ are expressions that as functions (but
not as expressions) satisfy equalities $g=\varphi_{k}$ and $h=\varphi_{m}$
for some (unknown) $k,m\in\mathbb{N}\cup\{\infty\}$. If $k=m$ then
$\lim_{(x,y)\rightarrow0}\frac{g(x,y)}{h(x,y)}=1$ otherwise the limit
does not exist. Suppose that no algorithms are available that would
simplify $g$ and $h$ to explicit polynomials. Computation of the
Taylor polynomials $T_{d}g$ and $T_{d}h$ for any finite $d$ does
not allow to decide whether $k$ and $m$ are finite, but greater
than $d/2$, or are infinite. Hence Problem \ref{WeakLimit} or \ref{StrongLimit}
cannot always be solved. Similarly, we cannot always decide whether
the zero of $h$ at $(0,0)$ is isolated.
\end{rem}

The algorithm we present here solves Problem \ref{WeakLimit} and
\ref{StrongLimit} if $h$ has an isolated zero at $0$. Otherwise
the algorithm does not terminate. Considering Remark \ref{limitation},
this is the best we can hope for with the given input specification. 

Let $U\subseteq\mathbb{R}^{n}$ be an open neighbourhood of $0$,
and let $\Vert u\Vert$ denote the Euclidean norm of $u\in\mathbb{R}^{n}$.
The algorithm is based on the following theorem, which is a special
case of the \L ojasiewicz inequality \cite{L}.

\begin{thm}
\label{LojIneq}Let $h\in\mathcal{A}(U)$ and suppose that $\{u\::\:\Vert u\Vert<\rho\wedge h(u)=0\}=\{0\}$
for some $\rho>0$. Then there exist positive constants $r$, $c$,
and $\alpha$ such that if $\Vert u\Vert<r$ then $\lvert h(u)\rvert\geq c\Vert u\Vert^{\alpha}$.
\end{thm}

\begin{lyxalgorithm}
\label{MLIM}(MLIM)\\
Input: $g\in\mathcal{A}(U)$, $h\in\mathcal{A}(U)\setminus\{0\}$,
such that $h(0)=0$.\\
Output: $\liminf_{u\rightarrow0}\frac{g(u)}{h(u)}$.

\begin{enumerate}
\item Set $d=2$.
\item Compute $\lim_{u\rightarrow0}\frac{u_{1}^{2d}+\cdots+u_{n}^{2d}}{T_{2d-1}h(u)}$.
\item If the limit does not exist or is not zero, set $d=d+1$ and go to
step $(2)$.
\item Return $\liminf_{u\rightarrow0}\frac{T_{2d-1}g(u)}{T_{2d-1}h(u)}$.
\end{enumerate}
\end{lyxalgorithm}

\begin{thm}
If $h$ has an isolated zero at $0$ then Algorithm \ref{MLIM} terminates
and returns 
\[
\liminf_{u\rightarrow0}\frac{g(u)}{h(u)}
\]
Otherwise the algorithm does not terminate.
\end{thm}

\begin{proof}
Suppose that $h$ has an isolated zero at $0$. By Theorem \ref{LojIneq}
there exist positive constants $r$, $c$, and $\alpha$ such that
if $\Vert u\Vert<r$ then $\lvert h(u)\rvert\geq c\Vert u\Vert^{\alpha}$.
To prove that Algorithm \ref{MLIM} terminates it suffices to show
that if $2d>\alpha$ then
\[
\lim_{u\rightarrow0}\frac{u_{1}^{2d}+\cdots+u_{n}^{2d}}{T_{2d-1}h(u)}=0
\]
Let $R_{2d-1}h(u)=h(u)-T_{2d-1}h(u)$. We have
\[
\frac{u_{1}^{2d}+\cdots+u_{n}^{2d}}{\lvert T_{2d-1}h(u)\rvert}\leq\frac{\Vert u\Vert^{2d}}{\lvert h(u)-R_{2d-1}h(u)\rvert}\leq\frac{1}{\lvert\frac{\lvert h(u)\rvert}{\Vert u\Vert^{2d}}-\frac{\lvert R_{2d-1}h(u)\rvert}{\Vert u\Vert^{2d}}\rvert}
\]
Since $R_{2d-1}h(u)$ is an analytic function whose Taylor series
does not contain terms of degree lower than $2d$, $\frac{\lvert R_{2d-1}h(u)\rvert}{\Vert u\Vert^{2d}}$
is bounded in a neighbourhood of $0$. Moreover, if $\Vert u\Vert<r$,
\[
\frac{\lvert h(u)\rvert}{\Vert u\Vert^{2d}}\geq\frac{c\Vert u\Vert^{\alpha}}{\Vert u\Vert^{2d}}=c\Vert u\Vert^{\alpha-2d}\underset{u\rightarrow0}{\longrightarrow}\infty
\]
hence 
\[
\lim_{u\rightarrow0}\frac{1}{\lvert\frac{\lvert h(u)\rvert}{\Vert u\Vert^{2d}}-\frac{\lvert R_{2d-1}h(u)\rvert}{\Vert u\Vert^{2d}}\rvert}=0
\]
which proves that
\[
\lim_{u\rightarrow0}\frac{u_{1}^{2d}+\cdots+u_{n}^{2d}}{T_{2d-1}h(u)}=0
\]
To show that Algorithm \ref{MLIM} returns $\liminf_{u\rightarrow0}\frac{g(u)}{h(u)}$
note that
\[
\frac{g(u)}{h(u)}=\frac{T_{2d-1}g(u)+R_{2d-1}g(u)}{T_{2d-1}h(u)+R_{2d-1}h(u)}=\frac{\frac{T_{2d-1}g(u)}{T_{2d-1}h(u)}+\frac{R_{2d-1}g(u)}{T_{2d-1}h(u)}}{1+\frac{R_{2d-1}h(u)}{T_{2d-1}h(u)}}
\]
We have 
\[
\frac{R_{2d-1}g(u)}{T_{2d-1}h(u)}=\frac{R_{2d-1}g(u)}{\Vert u\Vert^{2d}}\frac{\Vert u\Vert^{2d}}{T_{2d-1}h(u)}
\]
Since $R_{2d-1}g(u)$ is an analytic function whose Taylor series
does not contain terms of degree lower than $2d$, $\frac{\lvert R_{2d-1}g(u)\rvert}{\Vert u\Vert^{2d}}$
is bounded in a neighbourhood of $0$. Moreover, 
\[
\frac{\Vert u\Vert^{2d}}{\lvert T_{2d-1}h(u)\rvert}\leq n^{d}\frac{u_{1}^{2d}+\cdots+u_{n}^{2d}}{\lvert T_{2d-1}h(u)\rvert}\underset{u\rightarrow0}{\longrightarrow}0
\]
hence
\[
\lim_{u\rightarrow0}\frac{R_{2d-1}g(u)}{T_{2d-1}h(u)}=0
\]
Similarly
\[
\lim_{u\rightarrow0}\frac{R_{2d-1}h(u)}{T_{2d-1}h(u)}=0
\]
and therefore
\[
\liminf_{u\rightarrow0}\frac{g(u)}{h(u)}=\liminf_{u\rightarrow0}\frac{T_{2d-1}g(u)}{T_{2d-1}h(u)}
\]

Suppose now that the zero of $h$ at $0$ is not isolated. Let $d\geq2$.

If the zero of $T_{2d-1}h$ at $0$ is not isolated, then
\[
\frac{u_{1}^{2d}+\cdots+u_{n}^{2d}}{T_{2d-1}h(u)}
\]
attains arbitrarily large values in any neighbourhood of $0$, and
hence
\[
\lim_{u\rightarrow0}\frac{u_{1}^{2d}+\cdots+u_{n}^{2d}}{T_{2d-1}h(u)}
\]
does not exist or is not zero.

If the zero of $T_{2d-1}h$ at $0$ is isolated, then $\{u\::\:\Vert u\Vert\leq\rho\wedge T_{2d-1}h(u)=0\}=\{0\}$
for some $\rho>0$. Since $R_{2d-1}h(u)$ is an analytic function
whose Taylor series does not contain terms of degree lower than $2d$,
there exists $M>0$ such that $\frac{\lvert R_{2d-1}h(u)\rvert}{\Vert u\Vert^{2d}}\leq M$
for all $\Vert u\Vert\leq\rho$. Let $Z=\{u\::\:\Vert u\Vert\leq\rho\wedge h(u)=0\}$.
For $u\in Z\setminus\{0\}$ we have
\[
T_{2d-1}h(u)=-R_{2d-1}h(u)
\]
and hence
\[
\frac{u_{1}^{2d}+\cdots+u_{n}^{2d}}{\lvert T_{2d-1}h(u)\rvert}\geq n^{-d}\frac{\Vert u\Vert^{2d}}{\lvert R_{2d-1}h(u)\rvert}\geq n^{-d}M^{-1}>0
\]
Since $0$ is a limit point of $Z\setminus\{0\}$,
\[
\lim_{u\rightarrow0}\frac{u_{1}^{2d}+\cdots+u_{n}^{2d}}{T_{2d-1}h(u)}
\]
does not exist or is not zero. Therefore Algorithm \ref{MLIM} does
not terminate.
\end{proof}
\begin{rem}
Algorithm \ref{MLIM} can be adapted to compute the upper limit or
the limit of $f(u)$ by returning the upper limit or the limit of
$\frac{T_{2d-1}g(u)}{T_{2d-1}h(u)}$ in step $(4)$.
\end{rem}

\begin{example}
Algorithm \ref{MLIM} may not stop for the first $d$ such that $T_{2d-1}h$
has an isolated zero at $0$, even if $T_{2d-1}h=h$. Let $h(x,y)=(x^{n})^{2}+(x-y^{n})^{2}$
(cf. \cite{KS}, Example 1). Then $T_{2d-1}h=h$ for $d\geq n+1$.
However, since $h(t^{n},t)=t^{2n^{2}}$
\[
\lim_{u\rightarrow0}\frac{x^{2d}+y^{2d}}{T_{2d-1}h(x,y)}
\]
will not be zero for any $d\leq n^{2}$.
\end{example}

\begin{example}
The second part of the proof does need two cases, that is the zero
of $T_{2d-1}h$ at $0$ may be isolated even if the zero of $h$ at
$0$ is not isolated. Let $h(x,y)=y^{4}+(y-x^{2})^{2}-x^{6}-y^{6}$.
Then $T_{5}h=y^{4}+(y-x^{2})^{2}$ has an isolated zero at $0$. However,
the zero of $h$ at $0$ is not isolated. In a neighbourhood of zero
there are two analytic solutions of $h(x,y)=0$ with initial series
terms given by
\begin{eqnarray*}
y_{1}(x) & = & x^{2}-x^{3}+\frac{x^{5}}{2}-2x^{6}+\frac{25x^{7}}{8}+\cdots\\
y_{2}(x) & = & x^{2}+x^{3}-\frac{x^{5}}{2}-2x^{6}-\frac{25x^{7}}{8}+\cdots
\end{eqnarray*}
\end{example}

\begin{rem}
In some cases it is possible to detect that the zero of $h$ at $0$
is not isolated and terminate the algorithm. For instance if $h_{m}$
is the lowest degree nonzero form of the Taylor series of $h$ and
$h_{m}(a)>0$ and $h_{m}(b)<0$ for some $a,b\in\mathbb{R}^{n}$,
then $h(ta)>0$ and $h(tb)<0$ for $t\in(0,\epsilon)$ with some $\epsilon>0$,
and hence the zero of $h$ at $0$ is not isolated.

The number of limit computations in step $(2)$ can be reduced by
using fast negative criteria to decide that the zero of $T_{2d-1}h(u)$
at $0$ is not isolated.
\end{rem}

\section{Example}

Let us compute the lower limit and the upper limit of $\frac{g(x,y,z)}{h(x,y,z)}$
at $0$, where
\begin{eqnarray*}
g & = & \exp(\sin(x^{2}+y^{4}+z^{6}))-1\\
h & = & \sqrt{\cos(x)-\sin(y^{2})-z^{4}}-1
\end{eqnarray*}
For $d=2$ in step $(2)$ we have $T_{2d-1}h=\frac{-x^{2}-2y^{2}}{4}$.
Since $T_{2d-1}h(0,0,z)=0$, the limit 
\[
\lim_{(x,y,z)\rightarrow0}\frac{x^{4}+y^{4}+z^{4}}{T_{2d-1}h}
\]
does not exist, and so in step $(3)$ we set $d=3$ and go back to
step $(2)$. Now 
\[
T_{2d-1}h=\frac{-x^{4}-12x^{2}y^{2}-24x^{2}-12y^{4}-48y^{2}-48z^{4}}{96}
\]
and
\[
\lim_{(x,y,z)\rightarrow0}\frac{x^{6}+y^{6}+z^{6}}{T_{2d-1}h}=0
\]
hence we move on to step $(4)$. We have $T_{2d-1}g=\frac{x^{4}+2x^{2}+2y^{4}}{2}$
and the returned value is
\[
\liminf_{(x,y,z)\rightarrow0}\frac{g(x,y,z)}{h(x,y,z)}=\liminf_{(x,y,z)\rightarrow0}\frac{T_{2d-1}g(x,y,z)}{T_{2d-1}h(x,y,z)}=-4
\]
To find the upper limit we just need to compute the upper limit in
step $(4)$
\[
\limsup_{(x,y,z)\rightarrow0}\frac{g(x,y,z)}{h(x,y,z)}=\limsup_{(x,y,z)\rightarrow0}\frac{T_{2d-1}g(x,y,z)}{T_{2d-1}h(x,y,z)}=0
\]
To compute limits of rational functions Algorithm \ref{MLIM} can
use any of the algorithms described in \cite{S15}. We used the implementation
of these algorithms in \emph{Mathematica}. In this example methods
based on topological properties performed better. Algorithm 15 (TLIM)
took $77$ seconds when using Algorithm 14 (ZCQ2) and $323$ seconds
when using Algorithm 13 (ZCQ1). Methods based on optimization were
not able to complete the computation in $12$ hours. The most time-consuming
part of the computation is the limit in step $(2)$ with $d=3$.

%\bibliography{MeroFunLimit}

\begin{thebibliography}{10}

\bibitem{AKM}
P.~Alvandi, M.~Kazemi, and M.~Moreno Maza.
\newblock Computing limits of real multivariate rational functions.
\newblock In {\em Proceedings of the International Symposium on Symbolic and
  Algebraic Computation, ISSAC 2016}, pages 39--46. ACM, 2016.

\bibitem{CMV}
C.~Cadavid, S.~Molina, and J.~D. Velez.
\newblock Limits of quotients of bivariate real analytic functions.
\newblock {\em J. Symbolic Comp.}, 50:197--207, 2013.

\bibitem{G}
D.~Gruntz.
\newblock {\em On computing limits in a symbolic manipulation system}.
\newblock PhD thesis, ETH, 1996.

\bibitem{KS}
K.~Kurdyka and S.~Spodzieja.
\newblock Separation of real algebraic sets and the \l{}ojasiewicz exponent.
\newblock {\em Proceedings of the AMS}, 142(9):3089--3102, 2014.

\bibitem{L}
S.~\L{}ojasiewicz.
\newblock {\em Ensembles semi-analytiques}.
\newblock I.H.E.S., 1964.

\bibitem{SSH}
B.~Salvy and J.~Shackell.
\newblock Symbolic asymptotics: Multiseries of inverse functions.
\newblock {\em J. Symb. Comput.}, 27:543--563, 1999.

\bibitem{S15}
A.~Strzebo\'nski.
\newblock Comparison of cad-based methods for computation of rational function
  limits.
\newblock In {\em Proceedings of the International Symposium on Symbolic and
  Algebraic Computation, ISSAC 2018}, pages 375--382. ACM, 2018.

\bibitem{VHC}
J.~D. Velez, J.~P. Hernandez, and C.~A. Cadavid.
\newblock Limits of quotients of real polynomial functions of three variables,
  2015.
\newblock arXiv 1505.04121.

\bibitem{XZ}
S.~J. Xiao and G.~X. Zeng.
\newblock Determination of the limits for multivariate rational functions.
\newblock {\em Science China Mathematics}, 57(2):397--416, 2014.

\bibitem{XZZ}
S.~J. Xiao, X.~N. Zeng, and G.~X. Zeng.
\newblock Real valuations and the limits of multivariate rational functions.
\newblock {\em Journal of Algebra and Its Applications}, 14(5):1550--1567,
  2015.

\end{thebibliography}
\bibliographystyle{plain}

\end{document}